%% file: mlsp_2022.tex
\documentclass{article}
\usepackage{amsmath,graphicx,mlspconf}

\usepackage{mathtools} 
\usepackage{amsthm}
\usepackage{amsmath}
\usepackage{amsfonts} 
\usepackage{graphicx}
\usepackage{float}
\usepackage{balance}

\usepackage{caption}

\usepackage{algorithm}
\usepackage{algpseudocode}

\newcommand\independent{\protect\mathpalette{\protect\independenT}{\perp}}
\def\independenT#1#2{\mathrel{\rlap{$#1#2$}\mkern2mu{#1#2}}}

\newcommand{\rom}[1]{\uppercase\expandafter{\romannumeral #1\relax}}

\newtheorem{definition}{Definition}

\newtheorem{corollary}{Corollary}
\newtheorem{proposition}{Proposition}

\input{defns}

\copyrightnotice{U.S.\ Government work not protected by U.S.\ copyright}

\copyrightnotice{978-1-6654-8547-0/22/\$31.00 {\copyright}2022 Crown}

\copyrightnotice{978-1-6654-8547-0/22/\$31.00
{\copyright}2022 European Union}

\copyrightnotice{978-1-6654-8547-0/22//\$31.00 {\copyright}2022 IEEE}

\toappear{2022 IEEE International Workshop on Machine Learning for Signal Processing, Agu.\ 22--25, 2022, Xi'an, China}

\title{Lower Bounds on the Error Probability\\ for Invariant Causal Prediction}
\twoauthors{%
    Austin Goddard and Yu Xiang
}{
    University of Utah \\
    Department of Electrical Engineering \\
    Salt Lake City, UT, USA \\
    \{austin.goddard, yu.xiang\}@utah.edu
}{
   Ilya Soloveychik
}{
    Hebrew University of Jerusalem \\
    Department of Statistics \\
    Jerusalem, Israel \\
    soloveychik.ilya@mail.huji.ac.il
}
 
\begin{document}

\maketitle

\begin{abstract}
\vspace{-0.5em}
  It is common practice to collect observations of feature and response pairs from different environments. A natural question is how to identify features that have consistent prediction power across environments. The invariant causal prediction framework proposes to approach this problem through invariance, assuming a linear model that is invariant under different environments. In this work, we make an attempt to shed light on this framework by connecting it to the Gaussian multiple access channel problem. Specifically, we incorporate optimal code constructions and decoding methods to provide lower bounds on the error probability. We illustrate our findings by various simulation settings.
\end{abstract}

\vspace{-0.5em}
\section{Introduction}
\vspace{-0.5em}


Invariant causal prediction (\textsf{ICP})~\cite{peters2016causal} is a recently proposed framework on linear models for selecting features that are stable across different environments. It is motivated by the idea of invariance, which is often referred to as modularity, autonomy \cite{haavelmo1944probability, aldrich1989autonomy,hoover1990logic,pearl2000models, scholkopf2012causal}, or stability~\cite{dawid2010identifying,pearl2000models}. The main assumption for invariant prediction is that \emph{there exists a linear model invariant across environments}, with an unknown noise distribution and arbitrary dependence among predictors. Roughly speaking, it assumes for all environments $e\in\mathcal{E}$ there exists $Y^e =  X^e\gamma + Z^e$, where $Z^e$ is distributed according to an \emph{unknown} $F$ (that does not depend on $e$) and $Z^e$ is \emph{independent} of a set of true predictors $X^e_{S^*}$ indexed by set $S^*\subseteq\{1,..., m\}$ (with $|S^*|=k$). With a total of $n$ observations from different environments, the goal of \textsf{ICP} is to approximate $S^*$ in a computationally efficient manner.

There is a rich line of works on support recovery from an information-theoretic perspective~\cite{jin2011limits,scarlett2016limits} (and we only list the most relevant ones due to space limit). The \textsf{ICP} framework is different from the traditional support recovery settings in that (a) the set of predictors may not be unique (see Discussion in~\cite{peters2016causal}), (b) the number of potential predictors $m$ may not grow with number of observations $n$, and (c) it allows for arbitrary dependencies among features (while in support recovery settings, the measurement matrix is usually assumed to be \iid~\cite{jin2011limits,scarlett2016limits}; partly because the \iid codebook is optimal for the Gaussian multiple access channel when $m$ is exponential in $n$). 

In this work, we leverage information-theoretic techniques for Gaussian point-to-point and multiple access channels to shed light on the character of \textsf{ICP}. This connection enables us to use power constraints to differentiate environments and apply optimal codeword constructions and decoding methods to obtain lower bounds on the error probability. To make this connection possible, we consider the simplest setting that guarantees a unique support such that it becomes reasonable to discuss the lower bound on the error probability of support recovery; at the same time we try to keep the dependency between variables as general as possible. Specifically, we assume that both Gaussian noise and $k$ are known, and study cases when the channel gains are known or need to be estimated. We focus mainly on the zero-rate setting, where the number of potential predictors $m$ does not grow with the number of samples $n$~\footnote{To simplify notation, we skip settings when $m$ is polynomial in $n$ (or anything slower than exponential), which also belong to the zero-rate setup.}. In the zero-rate setting, the optimal code is known to be the  simplex code (detailed in Section~\ref{sec:simplex}), which we use to understand the impact of different environments on \textsf{ICP}. The analysis for the positive-rate case (when $m$ is exponential in $n$) is possible by adopting Fano's inequality~\cite{cover1999elements} and is discussed briefly in Section~\ref{sec::pos_rate}.

\vspace{-0.5em}
\section{Background and Problem Setting}
\vspace{-0.5em}
\subsection{Methods for Invariant Causal Prediction}
\vspace{-0.5em}

Consider a setting in which we have different experimental conditions and let $\mathcal{E}$ denote the index set of all possible environmental conditions. In each $e \in \mathcal{E}$, let the feature vector $X^e = (X_1^e,\dots, X^e_m)^\top\in\mathbb{R}^{m\times 1}$ and the response $Y^e\in \mathbb{R}$ form a joint distribution $(X^e,Y^e)$. For a set $S \subseteq \{1,\dots ,m\}$, $X_S^e$ denotes a random vector containing all variables $X_i^e$, $i \in S$. In this work, we focus on the simplest setting where there are only \emph{two environments}, i.e., $|\mathcal{E}|=2$. 

The main assumption for \textsf{ICP} is that there exist a vector of coefficients $\gamma = (\gamma_1,\dots,\gamma_m)^\top$ with support $S^* \coloneqq \{i:\gamma_i \neq 0 \}$ that, for both environments, satisfies: $X^e$  has an arbitrary distribution and
 \begin{equation}\label{equ:inv}
      Y^e =  X^e\gamma + Z^e, Z^e \sim F \text{ and } Z^e \independent X^e_{S^*},  
 \end{equation}
 where the (zero mean and finite variance) noise term $Z^e$ follows the same distribution $F$ across both environments~\footnote{We skip the intercept term in the model for simplicity of presentation.}. Variables in the vector $X_{S^*}^e$ are referred to as causal predictors, and the number of causal predictors (i.e., the cardinality of the support $S^*$) is $|S^*| = k$. Upon receiving $n$ observations of $(X^e,Y^e)$ and in each environment, the goal of \textsf{ICP} is to recover the support $S^*$.
 
Two \textsf{ICP} methods are proposed in~\cite{peters2016causal}, known both here and in the original work as Method~I and Method~II. The idea behind these and their variants (e.g.~~\cite{heinze2018invariant}) is to iterate over all subsets of variables $X_S$, $S \in \{1,\dots,m\}$ and test each subset for invariance. It is shown in~\cite{peters2016causal} that, with high probability, the resulting intersection of invariant sets will be a subset of $S^*$. Generally, only the test for invariance differentiates \textsf{ICP} algorithms. Method~I fits linear models in each environment and uses a test on regression coefficients to assert invariance. Method~II fits a linear model and, for each environment, tests the mean and variance of the residuals to determine invariance. 

\vspace{-0.5em}
\subsection{Error Exponent and Gaussian \textsf{MAC}} \label{subsec:brief_rev}
\vspace{-0.5em}

\noindent\underline{\bf Error exponent for the Gaussian channel}. We briefly review some classical results on error exponent from Shannon~\cite{shannon1959probability} for the point-to-point Gaussian channel. Consider a Gaussian point-to-point channel in which the sender has access to a codebook $C = \{c_1,c_2,\dots,c_m\}$ for $m$ messages, where $c_j \in \mathbb{R}^n$ and $m$ is the number of codewords in $C$. To transmit information, the sender first chooses a codeword and then sends the $i$-th element of the chosen codeword at transmission time $i$ as the input symbol $X_{i}$. We assume the \emph{peak energy constraint} $\sum_{i=1}^n x_i(l)^2  \leq nP$ for all messages $1\le l\le m$; this constraint allows us to model different environments in a natural manner. The receiver obtains $Y_i = h X_{i} + Z_i$, where $h$ is the channel gain, and the $Z_i$ are each assumed to be i.i.d. $N(0,\sigma^2_z)$. After $n$ transmissions, the receiver needs to determine which codeword in codebook $C$ was sent.

\vspace{-0.5em}
\begin{definition} [\bfseries Error exponent] \label{def:err_exp}
We define error exponent as the rate of decay for the error probability of the optimal sequence of $(m,n)$ codes. i.e.,
\begin{equation}
    E_m \coloneqq \limsup\limits_{n\rightarrow\infty} -\frac{1}{n}\ln P_e^*(m,n),
\end{equation}
where $P_e^*(m,n)$ denotes the best error probability over all $(m,n)$ codes. 
\end{definition}

 As shown by Shannon in~\cite[Equation~(82)]{shannon1959probability}, when $h = 1$, a lower bound on $P_e^*(m,n)$ for communicating using a codebook of $m$ codewords over a point-to-point channel is 
\begin{equation} \label{equ:opt_err}
    P_e^*(m,n) \ge  \frac{1}{2}\Phi\left(-\sqrt{\frac{m}{4(m-2)} \cdot\frac{nP}{2}}\right),
\end{equation}
where $\Phi(x)$ denotes the cumulative distribution function (CDF) of the standard Gaussian distribution. Accordingly, the error exponent is upper bounded by $ \frac{m}{4(m-1)}P$ (which follows from~\cite[Equation~(81)]{shannon1959probability} as (82) therein is a slightly loose bound). In fact, it is well-known that $E_{m} = \frac{m}{4(m-1)}P$ is optimal for zero-rate settings (which include the case of interest in this paper, i.e., when $m$ is fixed and does not grow with $n$), since it can be achieved using a regular simplex code on the sphere of radius $\sqrt{nP}$ along with minimum distance decoding. The zero-rate error exponent and simplex code play a fundamental role in communication problems such as the Gaussian channel with noisy feedback~\cite{xiang2013gaussian}. To the best of our knowledge, the optimal error exponent is unknown for the Gaussian multiple access channel under the zero-rate setting. 
\smallskip

\noindent\underline{\bf Gaussian \textsf{MAC} and support recovery}. 
One variant of the Gaussian multiple access channel (\textsf{MAC}) can be formulated similarly, where all the $k$ senders \emph{share the same codebook} $C = \{c_1,c_2,\dots,c_m\}$. The receiver obtains
\begin{equation}
  Y_i = h_1X_{1,i} + h_2X_{2,i} + \dots + h_kX_{k,i} +Z_i,
\end{equation}
where $h_l$ is the channel gain for sender $l$, and $X_{l,i}$ is the input symbol from sender $l$ at time $i$. From now on, we will simply refer to this common codebook setting as the Gaussian \textsf{MAC}. 

It is well-known that this Gaussian \textsf{MAC} setting and the support recovery problem in linear models are equivalent (apart from unknown channel gains $h_l$'s needing to be estimated properly~\cite{jin2011limits,scarlett2016limits}). We list the key similarities between these two and \textsf{ICP} as follows (see also details from~\cite{jin2011limits}): (1) The $k$ senders relate to the elements in $S^*$ or, similarly, the non-zero coefficients in $\gamma$; (2) The nonzero entries in $\gamma$ can be seen as channel gains; (3) The goal of codeword recovery for a Gaussian \textsf{MAC} is to determine the index of the sent codeword in the codebook; similarly, the support recovery problem must determine the support $S^*$ of the coefficients $\gamma$. 

\vspace{-0.5em}
\subsection{Connecting \textsf{ICP} with the Gaussian \textsf{MAC}}
\vspace{-0.5em}


Even though the three similarities mentioned above are shared between the Gaussian \textsf{MAC} and \textsf{ICP}, the two problems, unlike the support recovery and codeword recovery in a Gaussian \textsf{MAC}, are not equivalent. This is because (1) the notion of environment and invariance in \textsf{ICP}; (2) \textsf{ICP} is very general in its assumptions. e.g, the distributions $X^e_i$ and $Z^e$ are arbitrary. Careful assumptions must be made and resulting differences carefully thought-out and accounted for in order to shed light on \textsf{ICP} via the Gaussian \textsf{MAC}. We outline several differences mainly due to the generality of \textsf{ICP}. 

\vspace{-0.5em}
 \begin{enumerate}
 \setlength\itemsep{-0.25em}
     \item Codewords in a Gaussian MAC are often assumed to be independently distributed. In an \textsf{ICP} setting, however, dependencies between predictor variables are allowed.  
     \item In a Gaussian \textsf{MAC}, channel gains are known (e.g. from pilot sequences or feedback). In the case of ICP, the non-zero coefficients in $\gamma$ are actually unknown and need to be estimated. 
     \item In \textsf{ICP}, the number of invariant causal predictors, $k$, is not known. This is in contrast to a Gaussian MAC where the number of senders is always known.
     \item In a Gaussian \textsf{MAC}, the noise distribution is known to be \iid $N(0,\sigma_z^2)$. In \textsf{ICP}, the distribution of the noise can be arbitrary. 
 \end{enumerate}
 \vspace{-0.5em}
 
Because of the generality of \textsf{ICP}, $S^*$ might not be unique (see Discussion in~\cite{peters2016causal}), which is why most \textsf{ICP} methods iterate over all subsets then intersect accepted sets. This is a sharp contrast with both Gaussian \textsf{MAC} and \textsf{ICP}. Therefore, as a first attempt to connect Gaussian \textsf{MAC} and \textsf{ICP}, we need to restrict ourselves to some tractable settings where $S^*$ is unique. We study the following class of \textsf{ICP} problems: (1) \textsf{ICP} noise is known and distributed according to $\mathcal{N}(0,\sigma^2_z)$; (2) the non-zero coefficients in $\gamma$ are known (this can be relaxed in the next section), and (3) $k$ is known. 

These constraints guarantee a unique $S^*$, enabling us to lower bound the error probability of recovering $S^*$ by leveraging information-theoretic techniques. We will mainly focus on a natural setting for \textsf{ICP}, the zero-rate case, when $m$ is fixed. Then, we present analysis for the positive-rate case, when $m$ grows exponentially with $n$. It is noteworthy that even these two settings are highly non-trivial and we report only partial theoretical solutions with heuristic algorithms.





\vspace{-0.5em}
\section{Lower Bounds: Zero-Rate Case} \label{sec::zero_rate}
\vspace{-0.5em}

 In an \textsf{ICP} setting, it is not natural to assume the number of predictors need grow with the sample size. The more relevant setting is the zero-rate case, where $m$ is fixed and does not grow with $n$. Thus, this is the primary setting explored in this work. Additionally, from an algorithmic perspective, it would quickly become computationally infeasible to run \textsf{ICP} if $m$ grows exponentially as most \textsf{ICP} methods would iterate over an exponentially increasing number of subsets.
 
In order to leverage the lower bound for the Gaussian point-to-point channel in~\eqref{equ:opt_err}, we start with the setting where the channel gain is known and $k=1$ (the unknown channel gain setting is covered in Section~\ref{sec:unknown}). Suppose that~\eqref{equ:inv} holds with two environments ($n/2$ samples per environment),  $Z^e \sim \mathcal{N}(0,1)$, and a channel gain of one. We are now ready to present a lower bound on the error probability of recovering $S^*$ in \textsf{ICP}. 
\vspace{-0.3em}
\begin{proposition} \label{the:lower_b}
Suppose each predictor $X_i^e$ for $i \in \{1,\dots m\}$ is a deterministic vector in $\mathbb{R}^n$ that obeys the power constraint $\sum x_i^2 \leq nP/2 $ for environment $1$ and $\sum x_i^2 \leq n(P+d)/2 $ for environment $2$. 
The error probability of correctly recovering $S^*$ can be lower bounded as follows,
\begin{equation} \label{equ:ther1}
      \frac{1}{2}\Phi\left(-\sqrt{\frac{m}{4(m-2)}\cdot\frac{n(P+d/2)}{2}}\;\right).
\end{equation}
\end{proposition}

\begin{proof}
First, it is sufficient to lower bound $P_e^*(m,n)$ for the Gaussian point-to-point channel to obtain a lower bound in the \textsf{ICP} problem. This is because by definition, $P_e^*(m,n)$ corresponds to optimal (over all possible $(m,n)$ codes) error probability, which implies the best environments for support recovery in \textsf{ICP}. Furthermore, $P_e^*(m,n)$ is computed with known noise distribution and known channel gain, which are not given in \textsf{ICP}. 

Now, without loss of generality, assume n
$n$ is even. With additional power constraints on each codeword such that, for any given codeword, $\sum_{i=1}^{n/2} x_i^2 \leq nP/2$ and $\sum_{i=n/2+1}^{n} x_i^2 \leq n(P + d)/2$.
Thus, the peak energy constraint is $ \sum_{i=1}^n x_i^2  \leq n (P + d/2) $, and the result then follows directly from~\eqref{equ:opt_err}.
\end{proof}

Using the bound in~\eqref{equ:ther1}, we compare the performance of existing methods (Methods I and II) to the lower bound on the probability of error for \textsf{ICP}. The setup for this comparison is as follows. Each predictor in environment one follows a uniform distribution such that $X_i^{e_1} \sim U[0,\sqrt{P}]$. Predictors in environment two are distributed such that $X_i^{e_2} \sim U[0, \sqrt{P+d}]$. Sampling from distributions such that these ensures the constraints in Proposition~\ref{the:lower_b} are met. Unless otherwise noted, the probability of error is estimated by averaging the results over $1000$ instances for all simulated experiments.  

Results from the comparison can be seen in Figure \ref{fig:exp1}. Given the arguments in the proof of Proposition~\ref{the:lower_b}, it is not surprising to see that Method~I and Method~II behave sub-optimally. 
In the following section, we discuss the optimal codes and optimal procedures for recovering $S^*$ such that the lower bound in~\eqref{equ:ther1} is almost achieved (as~\cite[Equation~(82)]{shannon1959probability} is a slightly loose bound).

\begin{figure}
\centerline{\includegraphics[scale=.3]{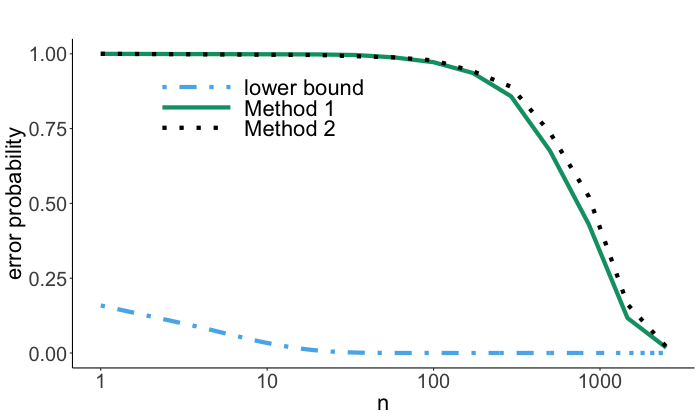}}
\caption{Robustness of linear \textsf{ICP} methods with respect to the lower bound in Proposition~\ref{the:lower_b}. $m=3$, $P=0.1$, and $d=1$.}
\vspace{-1.5em}
\label{fig:exp1}
\end{figure}

\vspace{-0.5em}
\subsection{Optimal Codes and Differences in Environment}
\label{sec:simplex}
\vspace{-0.5em}

In this section, we discuss codes used to achieve the optimal upper and lower bounds for \textsf{ICP} and how these codes relate to the distance between environments. In Proposition~\ref{the:lower_b}, we chose to model the differences in environment as codewords having different power constraints. This, however, is not the only interpretation for an environment in a \textsf{MAC} setting. For example, as we show in Section~\ref{sec::pos_rate}, environment can also be viewed as a shift in the mean of each codeword. In this section, we take a general view as to the definition of \emph{difference between environments}, and do so in terms of optimal and worst-case codes with the goal of understanding several general principles relating to environment. 

As mentioned, the bound in Proposition  \ref{the:lower_b} can be achieved using a regular simplex code on the sphere of radius $\sqrt{nP}$. In such a code, each codeword in $X$ satisfies the total power constraint $\sum_{i=1}^n X_i^2 = nP$ and is the same Euclidean distance from all other codewords. For example, when $m=3$, 
   \begin{align*}
    X_1 &= \sqrt{nP}\cdot[0,
           1,
           0,
           \cdots,
           0\;]^\top,\\
    X_2 &= \sqrt{nP}\cdot[-1/2,
           -\sqrt{3/2},
           0,
           \cdots,
           0\;]^\top,\\
    X_3 &=  \sqrt{nP}\cdot[1/2,
           -\sqrt{3/2},
           0,
           \cdots,
           0\;]^\top.
  \end{align*}

We refer to a codebook constructed in this way as $X_{\text{sim}}$.
Similarly, the worst-case code that can be constructed, assuming it satisfies the same total power constraint, is one such that all codewords are equal. That is,
    \begin{align*}
    X_i &= [\sqrt{P},
           \sqrt{P},
           \cdots,
           \sqrt{P}\;]^\top.
  \end{align*}
We refer to a codebook constructed in this manner as $X_{\text{unif}}$.

Since Gaussian noise is symmetric, the optimal decoding scheme, referred to as the minimum distance decoding (\textsf{MDD}), is such that the codeword geometrically closest to the received signal is the decoded codeword~\cite{shannon1959probability}. Recall, that minimum distance decoding is compatible with these designed codes in the setting in question, but may not be optimal in general. 

We now discuss the environments of codebooks $X_{\text{sim}}$ and $X_{\text{unif}}$. Since $X_{\text{sim}}$ and $X_{\text{unif}}$ represent best and worst-case codes, their environments must also represent the best and worst environments, respectively. Recall the first environment in the simplex codebook, relating to its first $n/2$ rows, contains at most $m(m-1)$ non-zero elements (assuming $n>2(m-1)$).  Note any simplex code, such as the example above for $m=3$, can have have $m(m-1)$ non-zero elements, simply by rotating the simplex. The second environment in the simplex codebook, relating to its second $n/2$ rows, contains only zero elements. The portions of each codeword belonging to environment one lie on a sphere of radius $\sqrt{nP}$ while the portions of each codeword for environment two all lie at the origin. Thus, distances between the portions of the codewords belonging to environment one and that of environment two can not get any larger without first increasing the power constraint, i.e., environments in the optimal code $X_{\text{sim}}$ are as different as possible. In the case of $X_{\text{unif}}$, since all elements of each codeword are equal, the worst-case code, $X_{\text{unif}}$, has identical environments. Thus, we see that optimal codes correspond to environments that are different while worst-case codes correspond to environments that are the same.

To further examine how distance between environments effects the probability of error, we form a codebook $X_{\text{inter}}$ whose difference in environments is larger than that of $X_{\text{unif}}$ but smaller than that of $X_{\text{sim}}$ using 
\begin{equation}\label{equ::bounded_area}
    	X_{\text{inter}} = aX_{\text{unif}} + (1-a)X_{\text{sim}},
\end{equation}
where $a$ ranges from 0 to 1. The value of $a$ can be seen as a measure of how different the two environments are, where $a=1$ being the least different and $a=0$, the most. An example of the transition between $X_{\text{sim}}$ and $X_{\text{unif}}$ can be seen in Figure~\ref{fig:trans1}. When $a=0$, we have the optimal case where $X_{\text{sim}}$ is paired with minimum distance decoding. When $a=1$, we have the optimal decoding scheme but the worst code, $X_{\text{unif}}$. In this case, the minimum distance decoding reverts to simply selecting, at random, any one of the codewords. Thus, for any $n$, the probability of error converges to $\frac{m-1}{m}$. Consequently, we come to see that no increase in sample size will improve the accuracy of \textsf{ICP} when data from environment one and environment two are equivalent, which is consistent with~\cite{peters2016causal}. Aside from the point where two environments are identical, we find the error probability consistently drops as the sample size $n$ increases, suggesting that, given a large enough $n$, one can achieve a probability of error equal to zero. In other words, given large enough $n$, one can always tell the minute differences in environment such that \textsf{ICP} can recover $S^*$. 

\begin{figure}
\centerline{\includegraphics[scale=.3]{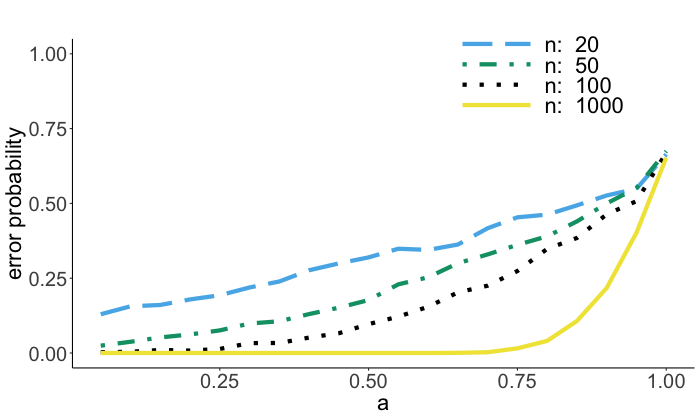}}
\caption{Error prob. using $X_{\text{inter}}$ with $m = 3$ and $P = 0.1$.}
\label{fig:trans1}
\vspace{-1.5em}
\end{figure}

\vspace{-0.5em}
\subsection{Unknown Channel Gains}
\label{sec:unknown}
\vspace{-0.5em}
Oftentimes in a Gaussian \textsf{MAC} setting, the channel gains are assumed to be known; while in \textsf{ICP}, the non-zero coefficients in $\gamma$ need to be estimated. Thus, we examine two methods for \textsf{ICP} in which the distribution of the noise is known but the coefficients in $\gamma$ are not. The first is a natural extension of minimum distance decoding where gains are estimated via ordinary least squares (\textsf{OLS}), while the second is a simple extension of a support recovery approach in \cite{jin2011limits}.

For case in which $k = 1$, a natural extension to minimum distance decoding would be to estimate the one non-zero coefficient in $\gamma$ using \textsf{OLS} for each possible $S^*$. Then, as in minimum distance decoding, the $S^*$ that produces the estimate closest to the received signal $Y^e$ is the accepted $\hat{S}^*$. The consequence of unknown coefficients in $\gamma$ is seen by comparing this procedure using the optimal code $X_{\text{sim}}$ to the lower bound in~\eqref{equ:opt_err} (see Figure~\ref{fig:unknown_hur_1}). 

\begin{figure}
\centerline{\includegraphics[scale=.30]{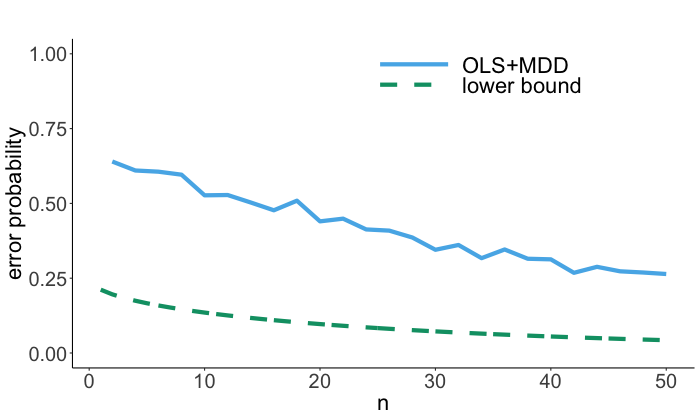}}
\caption{Error probability using $X_{\text{sim}}$ compared to the lower bound in~\eqref{equ:opt_err} where $m = 3$, $k=1$, and $P = 0.1$.}
\label{fig:unknown_hur_1}
\vspace{-1.5em}
\end{figure}

This procedure can be extended to the cases when $k > 1$ by examining the distance between the received signal $Y^e$ and its estimate $\hat{Y}_S^e = \sum_{i\in S} \hat{\gamma}_i X^e_{i}$ for all subsets $S \subseteq \{1,\dots,m\}$ such that $|S|=k$. The estimate $\hat{S}^*$ is the subset $S$ belonging to the $\hat{Y}_S^e$ closest to $Y^e$. We refer to this procedure, which is outlined in Algorithm~\ref{alg:ICP_decode}, as \textsf{OLS}+\textsf{MDD}. If we assume $k=1$, and that the gains are known (i.e., \textsf{OLS} is unnecessary), \textsf{OLS}+\textsf{MDD} is reduced to minimum distance decoding. While likely not optimum for $k \geq 1$, as minimum distance decoding is for $k=1$, \textsf{OLS}+\textsf{MDD} may provide a good estimate of $S^*$ under unknown gains. Similarly, the optimal code for $k=1$, $X_{\text{sim}}$, is likely not optimal for $k \geq 2$. However, we report results using $X_{\text{sim}}$ so as to compare results with the optimal $k=1$ case.
\vspace{-1em}
\begin{algorithm}
\caption{OLS+MDD}\label{alg:ICP_decode}
\hspace*{\algorithmicindent} \textbf{Input:} Received signal $Y^e$, codebook $X^e$, and number of causal predictors $k$

\hspace*{\algorithmicindent}
\textbf{Output:} $\hat{S}^*$
\begin{algorithmic}
\For{every subset of variables $X^e_S$ in $X^e$ such that $|S|=k$}
\State Estimate non-zero coefficients in $\gamma$ assuming $S^*=S$ using \textsf{OLS} estimates
\State Predict $\hat{Y}^e_S = \sum_{i\in S} \hat{\gamma}_i X^e_{i}$
\State Compute $R_S=|| Y^e-\hat{Y}^e_S||_2$
\EndFor
\State  Output: The $S$ associated with the smallest $R_S$
\end{algorithmic}
\end{algorithm}
\vspace{-1em}

Results for simulations done using Algorithm~\ref{alg:ICP_decode} and $X_{\text{sim}}$ can be seen in Figure~\ref{fig:unknown_hur_2}. As might be expected for $k \geq 2$, the probability of error is greater than that of $k=1$. One factor contributing to this increase is that the number of subsets grows with $k$. For example, when $k=1$ and $m=10$, there are only $10$ possibilities for $S^*$. When $k=5$, one must select $S^*$ from $252$ possible choices. Similarly, as $k$ grows, accurately estimating gains is challenging as it becomes more likely to find some combination of variables in $X_S$ that explain $Y$ (including ones that are not causal predictors of $Y$).

\begin{figure}
\centerline{\includegraphics[scale=.3]{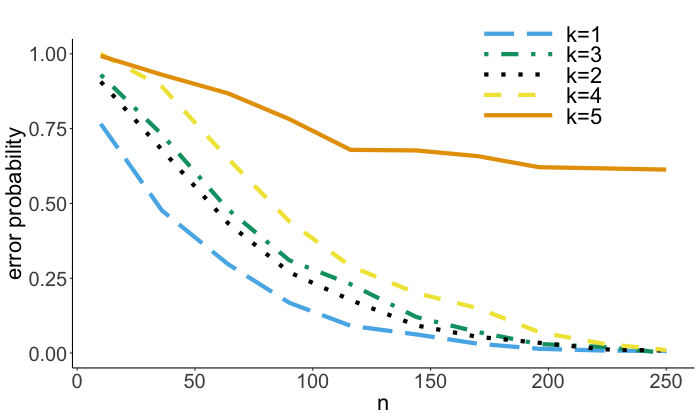}}
\caption{Error probability using $X_{\text{sim}}$ with $m = 3$ and $P = 0.1$.}
\label{fig:unknown_hur_2}
\vspace{-1.5em}
\end{figure}

As a second heuristic, for the $k=1$ case, we include a slight extension to a support recovery method proposed in \cite[Equations (19) and (20)]{jin2011limits} where codeword variance is simply replaced with sample variance. The procedure is as follows. For a codeword $j$, consider an estimate $\hat{\gamma}$ of $\gamma$ in which 
\begin{equation}
    \hat{\gamma} = \sqrt{\bigl(\hat{\sigma_j^2}\bigr)^{-1}\cdot\bigl| \frac{1}{n}||Y ||^2 - \sigma^2_z\bigr|}\;,
\end{equation}
where $\hat{\mu}_j := (1/n)\sum_{i=1}^n X_{i,j}$ is the sample mean and $\hat{\sigma_j^2}:=(1/(n-1))\sum_{i=1}^n (X_{i,j} - \hat{\mu}_j)^2$ is the sample variance. Then, declare that $\hat{S}^* = \{j\}$ if it is the unique index such that 
\begin{equation}
    \frac{1}{n}|| Y - (-1)^q\hat{\gamma}X_{\hat{S}^*} ||^2 \leq \sigma^2_z + \epsilon^2 \hat{\sigma_j^2}
\end{equation}
for all $j$ where $q$ is either $1$ or $2$ and $\epsilon$ is fixed such that it is greater than $0$. If for every $j \in \{1,\dots,k\}$ there is none that meets the above criteria, or if there are multiple, one is picked arbitrarily. Note that an extension to this method exists in \cite{jin2011limits} that allows for $k \geq 2$. However, due to the performance of the  \textsf{OLS}+\textsf{MDD} approach, we choose not to include any comparisons to this method.

When comparing the performance of this support recovery method, and the \textsf{OLS}+\textsf{MDD} method, we find \textsf{OLS}+\textsf{MDD} greatly outperforms the support recovery method when tested on $X_{\text{inter}}$ for $0 \leq a \leq 1$ (see Figure~\ref{fig:un_gain_3}). In fact, as $a$ grows, \textsf{OLS}+\textsf{MDD} behaves as \textsf{MDD} in the optimal case where the gains are assumed to be known. 

\begin{figure}
\centerline{\includegraphics[scale=.3]{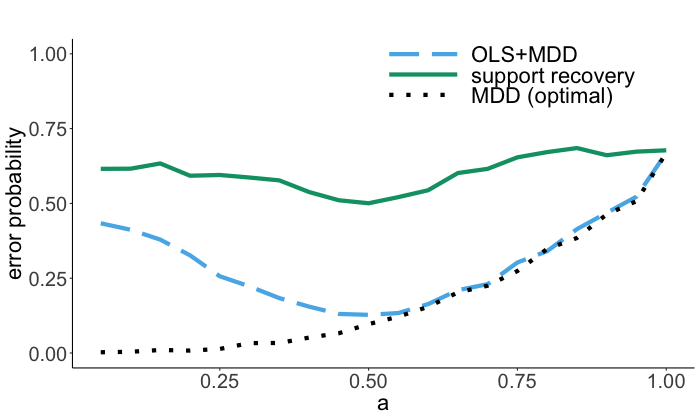}}
\caption{Error probability using $X_{\text{inter}}$ with $m = 3$, $k=1$, $P=0.1$, and $n=100$. \textsf{MDD} refers to the optimal minimum distance decoding approach where the gains are known.}
\label{fig:un_gain_3}
\vspace{-1.5em}
\end{figure}

\vspace{-0.5em}
\section{Lower Bounds: Positive-Rate Case} \label{sec::pos_rate}
\vspace{-0.5em}
As discussed previously, the bound in~\eqref{equ:ther1} refers to the zero rate case where $m$ does not grow with $n$. However, for cases in which it is permissible to allow $m$ to be large, and thus the rate to be positive, there exists several straightforward extensions derived from results in~\cite{jin2011limits}. This allows us to analyze a lower bound on the error probability for $k\geq1$. 

In particular, we examine two cases of primary interest, when differences in environment constitute shifts in mean and variance for $X^e$ being both random and fixed. For simplicity of notation, we assume the noise variance to be $1$.

\vspace{-0.25em}
\begin{corollary}\label{cor:fano1}
Consider an \textsf{ICP} setting with noise $Z \sim \mathcal{N}(0, 1)$. Suppose each predictor $X_i^e$ is independently distributed and has mean 0 and variance $1$ for environment one and mean $\mu_x$ and variance $1 + \sigma^2_d$ for environment two. For any $T\subseteq \{1,2,\dots k\}$, the error probability $P_e$ of recovering $S^*$ can be lower bounded by,
\begin{align*}
 P_e \geq \bigl( |T|&\log m - \frac{n}{4} \log  ( a + 1)((1+\sigma^2_d) a + 1)  -c_n \bigr)/b,
\end{align*}
where $a=\sum_{j\in T}\gamma_j^2$, $b=\log {m \choose k}$, and $c_{n} = \log k! + 1+n\log \bigl( m^{|T|} / \prod_{q=0}^{|T|-1} (m-(k-|T|)-q) \bigr)$.

\end{corollary}
\begin{proof}
It follows immediately from~\cite[Theorem 2]{jin2011limits} using Fano's inequality and we only outline the key differences from~~\cite[Equation~(41)]{jin2011limits}
\begin{alignat}{2}
    |T|\log m & \leq && P_e \log {m \choose k} + \frac{n}{4} \log \biggl(2\pi e \biggl(\sum_{j\in T}\gamma_j^2 + 1 \biggr) \notag\\ & && + \frac{n}{4} \log \biggl(2\pi e\biggl( \biggl( 1+\sigma^2_d \biggr) \sum_{j\in T}\gamma_j^2 + 1 \biggr) \biggr) \notag\\
    & && - \frac{1}{2}  \log k! + \log(2\pi e) + 1 + n\epsilon_{n}, \notag
\end{alignat}
where $\epsilon_{n} = \frac{1}{n}\log \bigl( m^{|T|} / \prod_{q=0}^{|T|-1} (m-(k-|T|)-q) \bigr)$. 
\end{proof}

Now by considering a deterministic $X^e$, the next result follows similarly and we omit the proof due to space limit. 

\vspace{-0.25em}
\begin{corollary}\label{cor:fano2}
In an \textsf{ICP} setting, suppose the set of  predictors $X^e$ is a deterministic matrix in $\mathbb{R}^{n\times m}$ where each column $X^e_j$ for $j \in \{1,\dots m\}$ obeys the energy constraint $\sum_i x_{i,j}^2 \leq 1 $. 
Assume that, after being generated, the second $n/2$ rows of $X^e$ are shifted by $\mu_d$. i.e., for $i > n/2+1$, $X^e_{i,j} = x_{i,j} + \mu_d$. Then, for any $T\subseteq \{1,2,\dots k\}$, the error probability $P_e$ of recovering $S^*$ can be lower bounded by, 
\begin{align*}
 P_e \geq \bigl[ |T|&\log m - \frac{n}{4} \log  \bigl( \bigl(a + \tau(m)\bigr) + 1\bigr) \\
 &- \frac{n}{4} \log\bigl( \eta(\mu_d) \bigl(a + \tau(m)\bigr) + 1\bigr) - c_n \bigr] \big/ b,
\end{align*}
where $(a,b,c_n)$ are defined the same as before in Corollary~\ref{cor:fano1},  $\tau(m) = \sum_{j\in T}\sum_{i \in T} \frac{\gamma_j\gamma_i}{m-1}$, and $\eta(m)=(1 + \mu^2_d + \frac{2\mu_d}{nm}\sum_{i=1}^n\sum_{\ell=1}^m x_{i,\ell} )$.
\end{corollary}
\vspace{-1.5em}

\bibliographystyle{unsrt}
\bibliography{refs}
\balance

\vspace{-0.5em}

\end{document}

%% file: defns.tex
\usepackage{xspace}
\usepackage{bbm}
\usepackage{mathrsfs}

\def\textiid{i.i.d.\@\xspace}
\newcommand\iid{\ifmmode\text{ i.i.d. } \else \textiid \fi}






%% file: mlsp_2022.bbl
\begin{thebibliography}{10}

\bibitem{peters2016causal}
Jonas Peters, Peter B{\"u}hlmann, and Nicolai Meinshausen.
\newblock Causal inference by using invariant prediction: identification and
  confidence intervals.
\newblock {\em Journal of the Royal Statistical Society. Series B (Statistical
  Methodology)}, pages 947--1012, 2016.

\bibitem{haavelmo1944probability}
Trygve Haavelmo.
\newblock The probability approach in econometrics.
\newblock {\em Econometrica: Journal of the Econometric Society}, pages
  iii--115, 1944.

\bibitem{aldrich1989autonomy}
John Aldrich.
\newblock Autonomy.
\newblock {\em Oxford Economic Papers}, 41(1):15--34, 1989.

\bibitem{hoover1990logic}
Kevin~D Hoover.
\newblock The logic of causal inference: Econometrics and the conditional
  analysis of causation.
\newblock {\em Economics \& Philosophy}, 6(2):207--234, 1990.

\bibitem{pearl2000models}
Judea Pearl.
\newblock Models, reasoning and inference.
\newblock {\em Cambridge, UK: CambridgeUniversityPress}, 19, 2000.

\bibitem{scholkopf2012causal}
Bernhard Sch{\"o}lkopf, Dominik Janzing, Jonas Peters, Eleni Sgouritsa, Kun
  Zhang, and Joris Mooij.
\newblock On causal and anticausal learning.
\newblock {\em arXiv preprint arXiv:1206.6471}, 2012.

\bibitem{dawid2010identifying}
A~Philip Dawid and Vanessa Didelez.
\newblock Identifying the consequences of dynamic treatment strategies: A
  decision-theoretic overview.
\newblock {\em Statistics Surveys}, 4:184--231, 2010.

\bibitem{jin2011limits}
Yuzhe Jin, Young-Han Kim, and Bhaskar~D Rao.
\newblock Limits on support recovery of sparse signals via multiple-access
  communication techniques.
\newblock {\em IEEE Transactions on Information Theory}, 57(12):7877--7892,
  2011.

\bibitem{scarlett2016limits}
Jonathan Scarlett and Volkan Cevher.
\newblock Limits on support recovery with probabilistic models: An
  information-theoretic framework.
\newblock {\em IEEE Transactions on Information Theory}, 63(1):593--620, 2016.

\bibitem{cover1999elements}
Thomas~M Cover.
\newblock {\em Elements of information theory}.
\newblock John Wiley \& Sons, 1999.

\bibitem{heinze2018invariant}
Christina Heinze-Deml, Jonas Peters, and Nicolai Meinshausen.
\newblock Invariant causal prediction for nonlinear models.
\newblock {\em Journal of Causal Inference}, 6(2), 2018.

\bibitem{shannon1959probability}
Claude~E Shannon.
\newblock Probability of error for optimal codes in a gaussian channel.
\newblock {\em Bell System Technical Journal}, 38(3):611--656, 1959.

\bibitem{xiang2013gaussian}
Yu~Xiang and Young-Han Kim.
\newblock Gaussian channel with noisy feedback and peak energy constraint.
\newblock {\em IEEE transactions on information theory}, 59(8):4746--4756,
  2013.

\end{thebibliography}
